\documentclass[twocolumn, prl, superscriptaddress,notitlepage]{revtex4-1}

\usepackage{graphicx}
\usepackage{amsmath}
\usepackage{amsthm}
\usepackage[usenames,dvipsnames]{color}
\usepackage{soul}
\usepackage{cancel}
\usepackage[
verbose,
colorlinks=true,
naturalnames=true,
linkcolor=blue,
]{hyperref}
\usepackage{braket}

\usepackage{amsfonts}

\hypersetup{colorlinks,linkcolor={blue},citecolor={blue},urlcolor={blue}}

\newtheorem{theorem}{Theorem}

\newtheorem{corollary}{Corollary}

\newcommand{\Tr}{{\rm Tr}}

\renewcommand{\vec}[1]{\boldsymbol{#1}} 

\begin{document}
	\title{Quantum Jarzynski Equality in Open Quantum Systems from the One-Time Measurement Scheme}

	\author{Akira Sone}
	\email{asone@lanl.gov}
	\affiliation{Theoretical Division, Los Alamos National Laboratory, Los Alamos, New Mexico 87545}
	\affiliation{Center for Nonlinear Studies, Los Alamos National Laboratory, Los Alamos, NM 87545}
	\affiliation{Research Laboratory of Electronics and Department of Nuclear Science and Engineering, Massachusetts Institute of Technology, Cambridge, Massachusetts 02139 }
	
	\author{Yi-Xiang Liu}
	\affiliation{Research Laboratory of Electronics and Department of Nuclear Science and Engineering, Massachusetts Institute of Technology, Cambridge, Massachusetts 02139 }
	
	\author{Paola Cappellaro}
	\email{pcappell@mit.edu}
	\affiliation{Research Laboratory of Electronics and Department of Nuclear Science and Engineering, Massachusetts Institute of Technology, Cambridge, Massachusetts 02139 }
	\affiliation{Department of Physics, Massachusetts Institute of Technology, Cambridge, Massachusetts 02139 }

	\begin{abstract}
		In open quantum systems, a clear distinction between work and heat is often challenging, and extending the quantum Jarzynski equality to systems evolving under general quantum channels beyond unitality remains an open problem in quantum thermodynamics. 
		{In this Letter, we introduce well-defined notions of \textit{guessed quantum heat} and \textit{guessed quantum work}, by exploiting the one-time measurement scheme, which only requires an initial energy measurement on the system alone. We derive a modified quantum Jarzynski equality and the principle of maximum work with respect to the guessed quantum work, which requires the knowledge of the system only.  We further show the significance of guessed quantum heat and work by linking them to the problem of quantum hypothesis testing.  }
	\end{abstract} 
	\maketitle

	\subparagraph{Introduction --}Understanding the laws of thermodynamics at the most fundamental level requires clarifying the thermodynamic properties of quantum systems, {and especially the contributions of coherence and correlations in the concept of work and heat are of fundamental interest  ~\cite{Binder19, DeffnerBook19, Santos19, Mohammady19, Scandi19,Bera17}}. 
	In quantum microscopic systems, fluctuations are inevitable; therefore, the laws of thermodynamics have to be given by taking into account the effects of these quantum fluctuations. 
	A powerful insight into fluctuations is provided by Jarzynski equality~\cite{Jarzynski97}, one of the few equalities in thermodynamics, which relates the fluctuating work in a finite-time,  nonequilibrium process with the equilibrium free energy difference:
	\begin{equation}
	\langle e^{-\beta W}\rangle=e^{-\beta\Delta F}\,,
	\label{eq:standard}
	\end{equation}
	Here $\beta=1/T$ is the inverse temperature {(we set the Boltzmann constant $k_B=1$)}; $W$ is the work; and $\Delta F_S$ is the equilibrium free energy difference  defined by the initial, $H_S\left(0\right)$, and final Hamiltonian, $H_S\left(t\right)$.
	The equality is independent of process details: the final state of the process does not have to be thermal, and the temperature could change. 
	Jarzynski equality can be also regarded as the generalization of the second law of thermodynamics, since through Jensen's inequality it yields the principle of maximum work: $\langle W\rangle \geq \Delta F$.

	The quantum version of the Jarzynski equality--- the quantum Jarzynski equality --- was developed by focusing on closed quantum systems in the two-time measurement scheme~\cite{ Tasaki00,Kurchan01}, which defines the work as the energy difference between the initial and final energy projection measurements in a single trajectory. 
	Jarzynski equality has been later extended to open quantum systems subject to dephasing process~\cite{Mukamel03}, unital maps~\cite{Rastegin13}, random projection measurements~\cite{Gherardini18,Thai18}, or feedback control~\cite{Sagawa10a, Sagawa10b}; it has been verified experimentally in numerous systems, such as biomolecular systems~\cite{Collin05}, trapped ions~\cite{Smith18, An15}, NV centers~\cite{Gomez19}, and NMR systems~\cite{Tiago14}.

	Despite this progress, a general formulation of the quantum Jarzynski equality for arbitrary open quantum systems is still lacking. This stems from the fundamental challenge that work and heat are not direct observables in quantum mechanics~\cite{Talkner07}: while in closed systems work can be simply identified with energy variations, in open quantum systems a clear distinction between work and heat is not always possible~\cite{Campisi11}. While some insight can be gained by theoretically assuming knowledge of the bath state~\cite{Alhambra16, Campisi09, Jarzynski04, Manzano18, Barra15}, in practice the bath cannot be measured. 
	One solution is to assume that a particular process does not involve heat exchange. For example, by assuming heat exchange to be absent in the dephasing process because there is no population decay, one can prove that the quantum Jarzynski equality has the standard form in Eq.~\eqref{eq:standard} \cite{Mukamel03,Smith18}. Similar results~\cite{Rastegin13} hold for unital maps (that is, identity-preserving maps), which describe only  processes that can be microscopically reversed by monitoring the bath with feedback~\cite{Buscemi05, Gregoratti03}.

	There have been several efforts to extend the quantum Jarzynski equality to nonunital maps~\cite{Morikuni17,Kafri12, Albash13, Goold15, Albash13, Rastegin14, Goold14}, by using the two-time measurement scheme. However, this either requires a measurement on the bath~\cite{Talkner09}, or it faces a fundamental issue~\cite{Marti17}, related to the loss of coherence in energy measurements. In open quantum system, the second energy measurement on the system unavoidably destroys system-bath correlations, making it impossible to distinguish  work and heat by energy measurements on the system alone, except for unitary or unital evolutions. In addition, the two-time measurement scheme neglects the information contribution due to the backaction of the second measurement~\cite{Deffner16}. {To improve on the results obtained by  the two-point measurement scheme, recent works have used  dynamic Bayesian networks~\cite{Micadei20} and Maggenau-Hill quasiprobability~\cite{Levy19}, or, for closed quantum systems, avoided completely the second measurement~\cite{Deffner16, Beyer20}.}

	In this Letter, we overcome these issues by introducing a novel definition of guessed quantum heat and work for general quantum channels, which lead to a quantum Jarzynski equality that takes into account system-bath correlations. 
	We employ the one-time measurement scheme developed in Ref.~\cite{Deffner16} for closed quantum systems. This protocol only requires us to measure the initial energy of the system (which is initially decoupled from the thermal bath) and to evaluate the expectation value of the difference between final and initial energy of the system by introducing the concept of ``best possible guess'' of the final state~\cite{Deffner16}. 
	Avoiding the final projective measurement of the energy provides a more precise description of the thermodynamic process than the traditional two-time measurement scheme, since it avoids the  backaction by the second measurement and the  ensuing information loss~\cite{Deffner16}. This protocol yields a modified quantum Jarzynski equality in terms of the information free energy~\cite{Deffner12, Still12, Deffner16, Sagawa15, Sivak12}, and a tighter bound on the second law of thermodynamics.

	Our main result is based on a generalization of the results in Ref.~\cite{Deffner16} to general quantum channels for open quantum systems in contact with a thermal bath. Inspired by the one-time measurement scheme, 
	we introduce well-defined notions of \textit{guessed quantum heat} and \textit{guessed quantum work} { that only require measurements on the system}. With these quantities, we can derive a modified quantum Jarzynski equality (see Theorem~\ref{theorem}) and further update the principle of maximum work (Corollary~\ref{corollary}). {Specifically, the bound in the principle of maximum work requires  knowledge of the system alone.} Not only the guessed quantum heat and work provide insights into the dynamics of general open quantum systems, {as we show with several examples~\cite{supp}}, but they acquire further operational meanings from their relationship to quantum hypothesis testing.

	\subparagraph{One-time measurement scheme --}
	We consider a composite system comprising the target system $(\mathcal{H}_S)$ and the bath $(\mathcal{H}_B)$, and assume we can only measure the system. Let $H_S(t)$ be the system Hamiltonian, which is time dependent, and $H_B$ the time-independent bath Hamiltonian. The total Hamiltonian, $H_{\text{tot}}(t)=H_S(t)\otimes\openone_B+\openone_S\otimes H_B+V(t)$,
	includes an interaction, $V(t)$, between system and bath  (we assume $V(t)\!=\!0$ for $t\leq0$).
	
	The initial state of the composite system is the product $\tau_S(0)\otimes\tau_B$ of thermal Gibbs states at $t=0$ for  system and bath,  $\tau_S(t)=e^{-\beta H_S(t)}/Z_S(t)$
	and $\tau_B=e^{-\beta H_B}/Z_B$.
	Here, $Z_A(t)$ are the partition functions, $Z_A(t)=\Tr\left[e^{-\beta H_A(t)}\right]$ for $A=S,B$. 
	The composite system evolves under a unitary operator $U_t$ as $U_t(\tau_S(0)\otimes\tau_B)U_t^{\dagger}$ which satisfies the usual Schr\"{o}dinger's equation $\partial_t U_t=-i H_{\text{tot}}(t)U_t$ (we set Plank constant $\hbar=1$).
	
	At time $t=0$, we measure the energy of the system alone. Suppose that we obtain a value $\epsilon$, corresponding to one of the eigenvalues of $H_S(0)$, with probability $e^{-\beta\epsilon}/Z_S(0)$. Then, the postmeasurement state of the system is the corresponding eigenstate: $\ket\epsilon\!\bra\epsilon$. Therefore, the evolved state of the system after the measurement is
	\begin{align*}
		\Phi_t \left(\ket\epsilon\!\bra\epsilon\right)\equiv\Tr_B\left[U_t(\ket\epsilon\!\bra\epsilon\otimes\tau_B)U_t^{\dagger}\right]\,,
	\end{align*}
	where $\Phi_t$ is a  {completely positive trace-preserving (CPTP)} map in $\mathcal{H}_S$. 
	This evolution includes contributions from heat exchange, because of the system coupling to the thermal bath, and from work due to the time dependence of the system Hamiltonian and to system-bath interaction, which exists even for time-independent Hamiltonians. It is however difficult to distinguish the two contributions, and, indeed, a measurement on the system alone would not be fully informative.

	After the evolution, we assume that we do not perform a final measurement, but still estimate the energy difference along a certain realization trajectory, $\Delta \tilde{E}(\epsilon)$, from the expectation value of the system Hamiltonian $H_S(t)$ with respect to $\Phi_t(\ket\epsilon\!\bra\epsilon)$:
	\begin{align*}
		\Delta \tilde{E}(\epsilon)=\Tr\left[H_S(t)\Phi_t(|\epsilon\rangle\langle\epsilon|)\right]-\epsilon\,.
	\end{align*}
	The probability distribution of the internal energy difference is given by
	\begin{align*}
		\tilde{P}(\Delta E)=\sum_{\epsilon}\frac{e^{-\beta\epsilon}}{Z_S(0)}\delta\Big(\Delta E-\Delta\tilde{E}(\epsilon)\Big)\,.
	\end{align*}
	This is a \textit{good} definition because it yields the correct expectation value  of the internal energy difference $\langle\Delta E\rangle$. Indeed, denoting with $\langle\cdots\rangle_{\tilde{P}}$ the average with respect to the distribution $\tilde{P}$, we have
	\begin{equation}
	\begin{split}
	\langle \Delta E\rangle_{\tilde{P}}=&\int\tilde{P}(\Delta E)\Delta Ed(\Delta E) \\
	=&\Tr\left[H_S(t)\Phi_t(\tau_S(0))\right]-\Tr\left[H_S(0)\tau_S(0)\right]\\
	\equiv&\langle \Delta E\rangle\,.
	\end{split}
	\label{eq:E}
	\end{equation}

	By using $\tilde{P}(\Delta E)$, we can calculate the averaged exponentiated internal energy difference:
	\begin{align*}
		\begin{split}
			\langle e^{-\beta \Delta E}\rangle_{\tilde{P}}&=\int \tilde{P}(\Delta E)e^{-\beta\Delta E}d(\Delta E)\\
			&=\frac{1}{Z_S(0)}\sum_{\epsilon}e^{-\beta\Tr\left[H_S(t)\Phi_t(|\epsilon\rangle\langle \epsilon|)\right]}\,.
		\end{split}
	\end{align*}
	We can interpret this expression by introducing a new partition function 
	\begin{align*}
		\tilde{Z}_S(t)\equiv \sum_{\epsilon}e^{-\beta\Tr\left[H_S(t)\Phi_t(|\epsilon\rangle\langle \epsilon|)\right]}\,,
	\end{align*}
	yielding
	\begin{equation}
	\langle e^{-\beta \Delta E}\rangle_{\tilde{P}}=\frac{\tilde{Z}_S(t)}{Z_S(0)}=e^{-\beta\Delta\tilde{F}_S}\,,
	\label{eq:Jarzynski1}
	\end{equation}
	where $\Delta \tilde{F}_S=\tilde{F}_S(t)-F_S(0)$
	is the difference between the initial, thermal equilibrium free energy, $F_S(0)=-\beta^{-1}\ln Z_S(0)$, and the equilibrium free energy corresponding to $\tilde{Z}_S(t)$, $\tilde{F}_S(t)=-\beta^{-1}\ln \tilde{Z}_S(t)$.
	We note that this relation has the form of a typical Jarzynski equality, linking the energy fluctuation to the free energy; however, to give this relation a physical meaning we need to further investigate the significance of $\tilde F_S(t)$ by linking this quantity to an effective state.

	\subparagraph{Guessed Quantum Heat \& Guessed Quantum Work --}
	Following Ref.~\cite{Deffner16}, we introduce the \textit{best possible
		guess} for the final system state. This thermal state, $\Theta_{SB}(t)$, can be found by maximizing the system-bath Von-Neumann entropy $\mathcal{S}_{SB}(t)=-\Tr\left[\Theta_{SB}(t)\ln\Theta_{SB}(t)\right]$, under the constraint of a fixed, average energy for the system alone, time-evolved after the one-time projective measurement. In other words, we apply the principle of maximum entropy~\cite{Jaynes57} to find the state with minimum information content, under the given constraints.
	The best possible guessed state can be given by
	\begin{align*}
		\Theta_{SB}(t)=\sum_{\epsilon}p(\epsilon)U_t(|\epsilon\rangle\langle\epsilon|\otimes\tau_B)U_t^{\dagger}\,,
	\end{align*}
	where the probabilities
	\begin{align*}
		p(\epsilon)=\frac{e^{-\beta\Tr\left[H_S(t)\Phi_t(|\epsilon\rangle\langle\epsilon|)\right]}}{\tilde{Z}_S(t)}\,,
	\end{align*}
	are found from entropy maximization under the constraint $E_S=\Tr\left[(H_S(t)\otimes\openone_B)\Theta_{SB}(t)\right]$ and that the postmeasurement state of the composite system after the initial energy measurement is given by $|\epsilon\rangle\langle\epsilon|\otimes\tau_B$, before evolving under $U_t$ (see \cite{supp}).

	We note that here we assumed an isothermal process for the composite system, as expected for a closed quantum system. Then, $\Theta_{SB}(t)$ can be seen as a thermal state at the initial temperature $\beta$, even if it is \textit{not} the thermal state of the composite system at time $t$, $\tau_S(t)\otimes\tau_B$.
	The difference can be quantified by their relative entropy $D\left[\Theta_{SB}(t)||\tau_S(t)\otimes\tau_B\right]\equiv\Tr\left[\Theta_{SB}(t)\ln\Theta_{SB}(t)\right]-\Tr\left[\Theta_{SB}(t)\ln(\tau_S(t)\otimes\tau_B)\right]$. 
	The relative entropy helps clarifying not only the thermodynamic contribution from the information difference of the states, but also an operational meaning of our results in terms of quantum hypothesis testing. 
	By defining 
	{
		\begin{align*}
			\langle \tilde{Q}\rangle_B \equiv \Tr\left[H_B\tau_B\right]-\Tr\left[(\openone_S\otimes H_B)\Theta_{SB}(t)\right]\,,  
		\end{align*}
	} we write $D$ as~\cite{supp}
	\begin{equation}
	D\left[\Theta_{SB}(t)||\tau_S(t)\otimes\tau_B\right]
	=-\ln\frac{\tilde{Z}_S(t)}{Z_S(t)}-\beta\langle \tilde{Q}\rangle_B\,.  
	\label{eq:DandHeat}
	\end{equation}
	Since $\langle \tilde{Q}\rangle_B$ represents the thermal bath energy loss, we can identify it as a kind of heat~\cite{Funo18}, that we  call ``guessed quantum heat'' as it arises from the definition of the best possible guessed state $\Theta_{SB}(t)$. 
	We can similarly introduce the notion of ``guessed quantum work'' $\tilde{W}$, based on the first law of thermodynamics: 
	\begin{equation}
	\tilde{W}\equiv \Delta E -\langle\tilde{Q}\rangle_B\,.
	\label{eq:guessedW}
	\end{equation}
	Then, we can obtain the following theorem:
	\begin{theorem}
		The quantum Jarzynski equality for the guessed quantum work is
		\begin{equation}
		\langle e^{-\beta\tilde{W}}\rangle_{\tilde{P}}=e^{-\beta \Delta F_{S}}e^{-D\left[\Theta_{SB}(t)||\tau_S(t)\otimes\tau_B\right]}\,.
		\label{eq:Jarzynski3}
		\end{equation}
		\label{theorem}
	\end{theorem}
	\begin{proof}
		From the definition of the equilibrium free energy, $F_S(t)=\beta^{-1}\ln Z_S(t)$, we can write
		$\tilde{F}_S(t)-F_S(t)=\langle\tilde{Q}\rangle_B+\beta^{-1}D\left[\Theta_{SB}(t)||\tau_S(t)\otimes\tau_B\right]$. 
		Defining $\Delta F_S=F_S(t)-F_S(0)$, we have
		\begin{align*}
			\Delta\tilde{F}_S =\Delta F_S+\langle\tilde{Q}\rangle_B+\beta^{-1}D\left[\Theta_{SB}(t)||\tau_S(t)\otimes\tau_B\right]\,, 
		\end{align*}
		and substituting into Eq.~\eqref{eq:Jarzynski1}, we obtain
		\begin{equation}
		\langle e^{-\beta\Delta E}\rangle_{\tilde{P}}=e^{-\beta\Delta F_S}e^{-\beta\langle\tilde{Q}\rangle_B}e^{-D\left[\Theta_{SB}(t)||\tau_S(t)\otimes\tau_B\right]}\,,
		\label{eq:Jarzynski2}
		\end{equation}
		which yields Eq.~\eqref{eq:Jarzynski3} using the definition of guessed quantum work in Eq.~\eqref{eq:guessedW}.
	\end{proof}
	Note that $\tilde{F}_S(t)$ plays the role of an information free energy~\cite{Deffner12, Still12, Deffner16, Sagawa15, Sivak12} computed with respect to the best possible guessed state $\Theta_{SB}(t)$.

	We verify Eq.~\eqref{eq:Jarzynski2} by considering several simple models in~\cite{supp}. We first discuss time-independent two-qubit interacting model, such as two-qubit dephasing.
	This model can be realized experimentally in two-qubit systems, such as nitrogen-vacancy (NV) centers in diamond~\cite{Liu19}, where $\mathcal{H}_S$ and $\mathcal{H}_B$ are the truncated electronic spin system and nuclear spin system associated with the NV center. We also consider an archetypal model of dephasing, the spin-boson model~\cite{Schlosshauer} without time dependence. In particular, by not assuming \textit{a priori} that dephasing precludes heat exchange, we find that we can define guessed quantum heat for dephasing maps, and thus guessed quantum work contains not only contributions from the Hamiltonian time dependence, but also from the interaction of system and bath.

	From Theorem~\ref{theorem}, we obtain the following corollary:
	\begin{corollary}[Principle of maximum guessed quantum work]
		The average of the guessed quantum work satisfies the following inequality:
		\begin{equation}
		\langle\tilde{W}\rangle \geq \Delta F_S+\beta^{-1}D\left[\tilde{\rho}_S(t)||\tau_S(t)\right]\,,
		\label{eq:2ndlaw}
		\end{equation}
		where $\tilde{\rho}_S(t)\equiv\Tr_B\left[\Theta_{SB}(t)\right]$.
		\label{corollary}
	\end{corollary}
	\begin{proof}
		Applying Jensen's inequality to Eq.~\eqref{eq:Jarzynski2}, and using the equivalence in Eq.~\eqref{eq:E}, from Eq.~\eqref{eq:guessedW}, we obtain
		\begin{equation}
		\langle \tilde{ W}\rangle\geq \Delta F_{S}+\beta^{-1}D\left[\Theta_{SB}(t)||\tau_S(t)\otimes\tau_B\right]\,.
		\label{eq:correlate}
		\end{equation}
		The monotonicity of the quantum relative entropy ~\cite{Nielsen00b} with respect to the partial trace leads to Eq.~\eqref{eq:2ndlaw} via 
		\begin{align*}
			D\left[\Theta_{SB}(t)||\tau_S(t)\otimes\tau_B\right]\geq D\left[\Tr_B\left[\Theta_{SB}(t)\right]||\tau_S(t)\right]\,.
		\end{align*}
		
	\end{proof}

	\subparagraph{Discussion --}
	The emergence of the guessed quantum heat and work can be understood as the results of system-bath correlations deriving from their interaction. As the one-time measurement does not erase such correlations, in contrast to the two-time measurement protocol, we are able to define and distinguish heat and work (their ``guessed'' values), {which are derived from the well-defined  guessed state}, even in cases such as dephasing where the two-time measurement protocol predicts no heat exchange. 
	
	Still, our results are consistent with well-known results for closed quantum systems. Since Eqs.~\eqref{eq:Jarzynski3} and \eqref{eq:2ndlaw} are generalizations of results in Ref.~\cite{Deffner16}, we can recover the  closed quantum system scenario by setting $V(t)=0$. Then, there is no energy exchange with the bath, i.e., no heat, and the guessed quantum work is simply the exact quantum work, given by the energy difference, $\langle \tilde{W}\rangle_{\tilde{P}}=\langle W\rangle=\langle \Delta E\rangle$. {Also, for the pure dephasing process, the guessed quantum work  coincides with the exact work, as we can see from examples in~\cite{supp}.}
	
	In contrast, for open quantum systems  Eqs.~\eqref{eq:Jarzynski3} and~\eqref{eq:correlate} introduce an additional thermodynamic contribution  to the work capacity,  given by the information difference between thermal and guessed state~\cite{Deffner16}, as quantified by the relative entropy. More precisely, the contribution arises from the difference between the product thermal state $\tau_S(t)\otimes\tau_B$ and the system-bath correlated state $\Theta_{SB}(t)$. This implies that system-bath correlations can increase the work capacity of the system.
	
	We indeed obtain a bound  for the principle of  maximum guessed quantum work 
	that importantly only requires knowledge of the system's state  (Eq.~\eqref{eq:2ndlaw}). 
	Avoiding measurements on the bath  is essential, as this  bound describes the maximum usable and extractable energy that the system can provide, which is of relevance for experiments and practical applications.
	
	To this goal, we were able to exploit the concept of ``guessed state'' not only to isolate the contribution from the measurement on the system, as done previously, but also to analyze the more realistic situation where the bath is  unmeasurable. In this scenario, then, $\Theta_{SB}(t)$ is a good effective state, because it can not only be estimated but it also {gives a bound to the guessed quantum work}, and similarly guessed quantum heat and work assume a well-defined meaning.

	Finally, we note that Eq.~(\ref{eq:Jarzynski3}) has operational meaning associated with the scaling of the quantum hypothesis testing from the quantum Stein's lemma~\cite{Ogawa00, Brandao10}. {The quantum relative entropy $D[\Theta_{SB}(t)||\tau_S(t)\otimes\tau_B]$ quantifies the distance between the guessed state $\Theta_{SB}(t)$ and the product Gibbs' state defined by the initial temperature and the final Hamiltonians of the system and bath $\tau_S(t)\otimes\tau_B$. This is associated with the type-I\!I error probability that the observation indicates the state to be $\Theta_{SB}(t)$ when the real state was $\tau_S(t)\otimes\tau_B$ (see ~\cite{supp} for details).}

	Assume that we prepare $n$ independent and identically distributed copies of $\Theta_{SB}(t)$ and $\tau_S(t)\otimes\tau_B$. {Here, $\Theta_{SB}(t)$ and  $\tau_S(t)\otimes\tau_B$ are seen as the null and alternative hypothesis, respectively.} Let us define $\mathcal{B}_n$ as the minimum type-I\!I error probability in quantum Stein's lemma that the true state is $(\tau_S(t)\otimes\tau_B)^{\otimes n}$ while the inferred state is $\Theta_{SB}^{\otimes n}(t)$. Then, in the limit of large $n$, we have 	
	\begin{equation}
	\lim_{n\to\infty}\frac{1}{n}\ln\left( \mathcal{B}_n\right)=-D\left[\Theta_{SB}(t)||\tau_S(t)\otimes\tau_B\right]\,.
	\label{eq:typeIIprobability}
	\end{equation}
	Relating the guessed quantum work $\tilde{W}$ (see Eq.~\eqref{eq:Jarzynski3}) with the type-I\!I probability $\mathcal{B}_n$,  
	\begin{equation}
	\langle e^{-\beta(\tilde{W}-\Delta F_S)}\rangle_{\tilde{P}}=	\lim_{n\to \infty}\left(\mathcal{B}_n\right)^{\frac{1}{n}}\,,
	\label{eq:operationalmeaning}
	\end{equation}
	{we can show that the guessed quantum work is asymptotically associated with the scaling of the quantum hypothesis testing when the true state is $\tau_S(t)\otimes\tau_B$ while the experimental result indicates $\Theta_{SB}(t)$.}

	In conclusion, we employ the one-time measurement scheme to derive a modified quantum Jarzynski equality and the principle of maximum quantum work in open quantum systems described by general quantum channels. We  demonstrate that the one-point measurement scheme enables  defining  heat and work with respect to the best possible guessed state, by introducing {well-defined} concepts of guessed quantum heat and guessed quantum work. Our work generalizes the  results obtained in  Ref.~\cite{Deffner16} for closed quantum systems, where guessed quantum work coincides with the exact quantum work. The extension to open quantum systems provides novel insights to the thermodynamics of both unital and generic quantum channels, by elucidating the role of correlations between system and bath in producing work and heat exchange, {as we illustrate in various examples in the Supplemental Material ~\cite{supp}}. Finally,  we also have shown the operational meaning of  guessed quantum work in terms of quantum hypothesis testing.  We expect that our results will contribute to a deeper understanding and further exploration of the role of work and heat in open quantum systems, as well as  quantum fluctuation theorems for general open quantum systems.

	
	\begin{acknowledgements}
		This work is in part supported by ARO MURI W911NF-11-1-0400 and MIT MIST-FVG. A. S. acknowledges Thomas G. Stockham Jr. Fellowship from MIT, and he is now supported by the U.S. Department of Energy, the Laboratory Directed Research and Development (LDRD) program and the Center for Nonlinear Studies at LANL. We offer our gratitude to Sebastian Deffner, Wojciech  H.  Zurek,  Yigit  Subasi,  Francesco  Caravelli, Stefano Gherardini, Stefano Ruffo, Andrea Trombettoni, Quntao  Zhuang,  and  Philippe  Faist  for helpful discussions.   We  are  also  grateful  to  Christopher  Jarzynski, Ryuji Takagi, Nicole Yunger Halpern, and Naoki Yamamoto for insightful discussions on stochastic thermodynamics.

	\end{acknowledgements}

	\bibliography{Biblio_Job}

	\newpage
	
	\onecolumngrid
	
	\setcounter{section}{0}
	\setcounter{figure}{0}
	\setcounter{corollary}{0}
	\setcounter{theorem}{0}
	\setcounter{condition}{0}
	\setcounter{proposition}{0}
	\setcounter{observation}{0}
	\setcounter{problem}{0}

	\onecolumngrid
	
	\subsection*{\Large{Supplementary Material for ``Quantum Jarzynski equality in open quantum systems from the one-time measurement scheme"} }
	\appendix
	
	\section{1. Best possible guessed state}
	
	We introduced in the main text the concept of ``guessed state". Here we show how to derive its expression following the principle of maximum entropy and the constraints imposed by the one-time measurement protocol.
	
	Initially, the system and the bath is decoupled, and the postmeasurement state of the composite system after the initial measurement is given by $|\epsilon\rangle\langle\epsilon|\otimes\tau_B$; therefore, we get a set of states after the unitary evolution $\{U_t(|\epsilon\rangle\langle\epsilon|\otimes\tau_B)U_t^{\dagger}\}_{\epsilon}$. These states are distributed based on the probability distribution $\{p(\epsilon)\}_{\epsilon}$, so that we can write the final state induced by the initial measurement as $\Theta_{SB}(t)=\sum_{\epsilon}p(\epsilon)U_t(|\epsilon\rangle\langle\epsilon|\otimes\tau_B)U_t^{\dagger}$. Then, we consider the following optimization problem. 
	
	Given a state $\Theta_{SB}(t)$:
	\begin{align*}
		\Theta_{SB}(t)=\sum_{\epsilon}p(\epsilon)U_t(|\epsilon\rangle\langle\epsilon|\otimes\tau_B)U_t^{\dagger}\,,
	\end{align*}
	let us consider the probability distribution $\{p(\epsilon)\}_{\epsilon}$ maximizing the Von-Neumann entropy $\mathcal{S}_{SB}(t)=-\Tr[\Theta_{SB}(t)\ln\Theta_{SB}(t)]$
	under the condition that 
	\begin{align*}
		\begin{split}
			&\Tr[\Theta_{SB}(t)]=1 \\
			&E_S=\Tr[(H_S(t)\otimes\openone_B)\Theta_{SB}(t)]\,,
		\end{split}
	\end{align*}
	so that
	\begin{align*}
		\begin{split}
			&\delta \Tr[\Theta_{SB}(t)]=\sum_{\epsilon}\delta p(\epsilon)=0\\
			&\delta E_S=\delta \Tr[(H_S(t)\otimes\openone_B)\Theta_{SB}(t)]=\sum_{\epsilon}\delta p(\epsilon)\Tr[(H_S(t)\otimes\openone_B)U_t(|\epsilon\rangle\langle\epsilon|\otimes\tau_B)U_t^{\dagger}]=\sum_{\epsilon}\delta p(\epsilon)\Tr[H_S(t)\Phi_t(|\epsilon\rangle\langle\epsilon|)]\,.
		\end{split}
	\end{align*}
	Here, note we only consider $H_S(t)$ because we assume that one can only measure the energy of the system. 
	Explicitly, $\Theta_{SB}(t)$ can be given by
	\begin{align*}
		\Theta_{SB}(t)=\sum_{\epsilon,q}p(\epsilon)\frac{1}{Z_B}e^{-\beta q}U_t|\epsilon,q\rangle\langle \epsilon,q|U_t^{\dagger}\,.
	\end{align*}
	Therefore, 
	\begin{align*}
		\delta \mathcal{S}_{SB}&=-\delta\Tr[\Theta_{SB}(t)\ln\Theta_{SB}(t)]\\
		&=-\sum_{\epsilon,q}\frac{e^{-\beta q}}{Z_B}\delta p(\epsilon)\Big(\ln p(\epsilon)-\ln Z_B-\beta q\Big)\\
		&=-\sum_{\epsilon}\delta p(\epsilon)\Big(\ln p(\epsilon)-\ln Z_B-\beta\Tr[H_B\tau_B]\Big)\,.
	\end{align*}

	By using the optimization method of Lagrange multipliers with constraints, we have: 
	\begin{align*}
		\delta\Big(\mathcal{S}_{SB}-\alpha E_S-\gamma\Big)&=
		-\sum_{\epsilon}\delta p(\epsilon)\Big(\ln p(\epsilon)-\beta\Tr[\tau_B H_B]-\ln Z_B
		+\alpha\Tr[H_S(t)\Phi_t(|\epsilon\rangle\epsilon|)]+\gamma+1\Big)\,.
	\end{align*}
	For any $\delta p(\epsilon)$, this has to be valid so that each term has to be independently 0. Therefore, 
	\begin{align*}
		\ln p(\epsilon)-\beta\Tr[H_B\tau_B]-\ln Z_B+\alpha\Tr[H_S(t)\Phi_t(|\epsilon\rangle\langle\epsilon|)]+\gamma+1=0\,,
	\end{align*}
	{so that we can obtain 
		\begin{align*}
			p_{\alpha}(\epsilon)\propto e^{-\alpha\Tr[H_S(t)\Phi_t(|\epsilon\rangle\langle\epsilon|)]}\,,
		\end{align*}
		where we put subscript $\alpha$ as  $p_{\alpha}(\epsilon)$ in order to emphasize the dependence of $p(\epsilon)$ on the parameter $\alpha$. Here, note that we can choose any $\alpha$, and we could have infinite numbers of guessed states. As the \textit{best} guessed state, since we do not know the final temperature, it is reasonable for us to choose $\alpha=\beta$.
		Since we have $\sum_{\epsilon}p_{\beta}(\epsilon)=1$, we can write
		\begin{align*}
			p_{\beta}(\epsilon)=\frac{e^{-\beta\Tr[H_S(t)\Phi_t(|\epsilon\rangle\langle\epsilon|)]}}{\tilde{Z}_S(t)}\,,
	\end{align*}}
	where $\tilde{Z}_S(t) =\sum_\epsilon e^{-\beta\Tr[H_S(t)\Phi_t(|\epsilon\rangle\langle\epsilon|)]}$.
	This means that the best possible guess of the thermal state of the composite system, which rises from the one-time measurement scheme, can be given by 
	\begin{align*}
		\Theta_{SB}(t)=\sum_{\epsilon}\frac{e^{-\beta\Tr[H_S(t)\Phi_t(|\epsilon\rangle\langle\epsilon|)]}}{\tilde{Z}_S(t)}U_t(|\epsilon\rangle\langle\epsilon|\otimes\tau_B)U_t^{\dagger}\,.
	\end{align*}

	\section{2. guessed quantum heat and Relative Entropy}
	\label{app:guessedheat}
	We introduce the guessed quantum heat when providing an explicit relationship between the relative entropy $D\left[\Theta_{SB}(t)||\tau_S(t)\otimes\tau_B\right]$ and the free energies, Eq.~\eqref{eq:DandHeat} of the main text. Here we provide an explicit proof  of this result.
	\begin{proof}
		First, let us calculate $\Tr\left[\Theta_{SB}(t)\ln\Theta_{SB}(t)\right]$.
		Since
		\begin{align*}
			\begin{split}
				\Theta_{SB}(t)&=\frac{1}{\tilde{Z}_S(t)}\sum_{\epsilon}e^{-\beta\Tr\left[H_S(t)\Phi_t(|\epsilon\rangle\langle\epsilon|)\right]}U_t(|\epsilon\rangle\langle\epsilon|\otimes\tau_B)U_t^{\dagger}\\
				&=\frac{1}{\tilde{Z}_S(t)}\frac{1}{Z_B}\sum_{\epsilon,q}e^{-\beta\Tr\left[H_S(t)\Phi_t(|\epsilon\rangle\langle\epsilon|)\right]}e^{-\beta q}U_t|\epsilon,q\rangle\langle\epsilon,q|U_t^{\dagger}\\
				&=\frac{1}{\tilde{Z}_S(t)}\frac{1}{Z_B}\sum_{\epsilon,q}e^{-\beta\Tr\left[(H_S(t)\otimes \openone_B)U_t(|\epsilon\rangle\langle\epsilon|\otimes\tau_B)U_t^{\dagger}\right]}e^{-\beta q}U_t|\epsilon,q\rangle\langle\epsilon,q|U_t^{\dagger}\,,
			\end{split}
		\end{align*}
		where we use the relation
		\begin{align*}
			\Tr\left[H_S(t)\Phi_t(|\epsilon\rangle\langle\epsilon|)\right]=\Tr\left[(H_S(t)\otimes \openone_B)U_t(|\epsilon\rangle\langle\epsilon|\otimes\tau_B)U_t^{\dagger}\right]\,.
		\end{align*}
		Therefore, we can obtain
		\begin{align*}
			\ln\Theta_{SB}(t)=-\ln\tilde{Z}_S(t)-\ln Z_B-\beta\sum_{\epsilon,q}\Big(\Tr\left[(H_S(t)\otimes\openone_B)U_t(|\epsilon\rangle\langle\epsilon|\otimes\tau_B)U_t^{\dagger}\right]+q\Big)U_t|\epsilon,q\rangle\langle\epsilon,q|U_t^{\dagger}\,.
		\end{align*}
		Then, we have
		\begin{align*}
			\Tr\left[\Theta_{SB}(t)\ln\Theta_{SB}(t)\right]=&-\ln\tilde{Z}_S(t)-\ln Z_B\\
			&-\beta\sum_{\epsilon,q}\Big(\Tr\left[(H_S(t)\otimes\openone_B)U_t(|\epsilon\rangle\langle\epsilon|\otimes\tau_B)U_t^{\dagger}\right]+q\Big)\cdot\frac{e^{-\beta\Tr\left[H_S(t)\Phi_t(|\epsilon\rangle\langle\epsilon|)\right]}}{\tilde{Z}_S(t)}\cdot\frac{e^{-\beta q}}{Z_B}\\
			=&-\ln\tilde{Z}_S(t)-\ln Z_B\\
			&-\beta\Tr\left[(H_S(t)\otimes\openone_B)\frac{1}{\tilde{Z}_S(t)}\sum_{\epsilon}e^{-\beta\Tr\left[H_S(t)\Phi_t(|\epsilon\rangle\langle\epsilon|)\right]}U_t(|\epsilon\rangle\langle\epsilon|\otimes\tau_B)U_t^{\dagger}\right]-\beta\sum_{q}\frac{e^{-\beta q}}{Z_B}q\\
			=&-\ln\tilde{Z}_S(t)-\ln Z_B-\beta\Tr\left[(H_S(t)\otimes\openone_B)\Theta_{SB}(t)\right]-\beta\Tr\left[H_B\tau_B\right]\,.
		\end{align*}
		Let us calculate $\Tr\left[\Theta_{SB}(t)\ln(\tau_S(t)\otimes\tau_B)\right]$. 
		Since
		\begin{align*}
			\tau_S(t)\otimes\tau_B=\frac{e^{-\beta H_S(t)}}{Z_S(t)}\otimes\frac{e^{-\beta H_B}}{Z_B}=\frac{1}{Z_S(t)Z_B}e^{-\beta(H_S(t)\otimes\openone_B+\openone_S\otimes H_B)}\,,
		\end{align*}
		we have
		\begin{align*}
			\Tr\left[\Theta_{SB}(t)\ln(\tau_S(t)\otimes\tau_B)\right]=-\ln Z_S(t)-\ln Z_B-\beta\Tr\left[(H_S(t)\otimes\openone_B)\Theta_{SB}(t)\right]-\beta \Tr\left[(\openone_S\otimes H_B)\Theta_{SB}(t)\right]\,.
		\end{align*}
		Therefore, the quantum relative entropy becomes
		\begin{align*}
			\begin{split}
				D\left[\Theta_{SB}(t)||\tau_S(t)\otimes\tau_B\right]&=\Tr\left[\Theta_{SB}(t)\ln\Theta_{SB}(t)\right]-\Tr\left[\Theta_{SB}(t)\ln(\tau_S(t)\otimes\tau_B)\right]\\
				&=-\ln\frac{\tilde{Z}_S(t)}{ Z_S(t)}-\beta\Big(\Tr\left[H_B\tau_B\right]-\Tr\left[(\openone_S\otimes H_B)\Theta_{SB}(t)\right]\Big)\,.
			\end{split}
		\end{align*}
		
	\end{proof}

	\section{3. Recovery of the closed-system case}
	We remark that our results are consistent with previous results obtained in the case of closed quantum system~\cite{Deffner16}. In closed quantum systems, there is no coupling to the bath, and the unitary evolution $U_t$ can be given by $U_t=\mathcal T\left[e^{-i\int dt H_S(t)}\right]\otimes e^{-iH_B t}$. Then, there is no energy loss to/from the bath, i.e., no heat, and the guessed quantum work is simply the exact quantum work, given by the energy difference, $\langle \tilde{W}\rangle=\langle W\rangle=\langle\Delta E\rangle$. The relative entropy, $D\left[\Theta_{SB}(t)||\tau_S(t)\otimes\tau_B\right]$, reduces to 
	\begin{align*}
		D\left[\Theta_{SB}(t)||\tau_S(t)\otimes\tau_B\right]=D\left[\tilde{\rho}_S(t)||\tau_S(t)\right]=\frac{\tilde{Z}_S(t)}{Z_S(t)}\,,
	\end{align*}
	where $\tilde{\rho}_S(t) = \textrm{Tr}_B(\Theta_{SB}(t))$ can be given explicitly by 
	\begin{align*}
		\tilde{\rho}_S(t) = \sum_{\epsilon}\frac{e^{-\beta\Tr\left[H_S(t)U_t^{(S)}|\epsilon\rangle\langle\epsilon|U_t^{(S)\dagger}\right]}}{\tilde{Z}_S(t)}U_t^{(S)}|\epsilon\rangle\langle\epsilon|U_t^{(S)\dagger}\,,
	\end{align*}
	where we define $U_t^{(S)}\equiv \mathcal{T}\left[e^{-i\int dt H_S(t)}\right]$.
	This is the close-system best possible guessed state as in Ref.~\cite{Deffner16}. In the absence of heat, the derived quantum Jarzynski equality and the maximum work in reduce to the main results of Ref.~\cite{Deffner16}:
	\begin{align*}
		\langle e^{-\beta W}\rangle_{\tilde{P}}=e^{-\beta\Delta F_S}e^{-D\left[\tilde{\rho}_S(t)||\tau_S(t)\right]}\,,
	\end{align*}
	and 
	\begin{align*}
		\langle W\rangle \geq \Delta F_S+\beta^{-1}D\left[\tilde{\rho}_S(t)||\tau_S(t)\right]\,.
	\end{align*}

	\section{4. Examples}

	We can further understand our main results by verifying our derived quantum Jarzynski equality:
	\begin{equation}
	\langle e^{-\beta \Delta E}\rangle_{\tilde{P}}=e^{-D\left[\Theta_{SB}(t)||\tau_S(t)\otimes\tau_B\right]-\beta\langle\tilde{Q}\rangle_B}\,
	\label{eq:JarzynskiSupp}
	\end{equation}
	with two toy models with different size of baths such as two-qubit dephasing and spin-boson model with time-independent Hamiltonian.
	
	In the following, $\vec{\sigma}_j=(\sigma_j^x, \sigma_j^y, \sigma_j^z)$ denotes the Pauli matrices for $j$-th spin, and $a_k$ ($a_k^{\dagger}$) is the annihilation (creation) operator of the $k$-th bosonic mode.

	The following results indicate that the system-bath interaction results in the guessed quantum work even in the composite systems characterized by the time-independent Hamiltonian.

	\subsubsection{4-1. Two-qubit dephasing model }
	Let us consider a single spin-1/2 system ($\mathcal{H}_S$) coupled to a single spin-1/2 bath ($\mathcal{H}_B$). For simplicity, let us consider a time-independent system Hamiltonian so that $\Delta F_S=0$. Here, $\vec{\sigma}_j=(\sigma_j^x, \sigma_j^y, \sigma_j^z)~(j=S,B)$ denotes the Pauli matrices for $j$-th spin.
	
	Let us consider $\sigma_S^z   \sigma_B^x$ coupling between system and bath. The total Hamiltonian becomes 
	\begin{align*}
		H= \omega_S \sigma_S^{z}+\omega_{B}\sigma^{z}_B + J \sigma_S^z   \sigma_B^x\,,
	\end{align*}
	where $J$ is the coupling strength. This simple two-qubit system models a  dephasing process for the system, as populations are preserved while coherences (initially) decay, {i.e. $\Phi_t(|\epsilon\rangle\langle\epsilon|)=|\epsilon\rangle\langle\epsilon|$, and in this case, the guessed state coincides with the exact state, i.e. $\Theta_{SB}(t)=U_t\left(\tau_S(0)\otimes\tau_B\right)U_t^{\dagger}$.} The system energy is thus conserved and we have $\langle e^{-\beta\Delta E}\rangle_{\tilde{P}} =1$.  In contrast, the backaction of the system evolution onto the bath leads to a change in energy of the bath itself,  and the guessed quantum heat and work can be given by
	
	\begin{align*}
		\langle \tilde{Q}\rangle_B =-\langle \tilde{W}\rangle=-\frac{2J^2 \omega_B \tanh(\beta \omega_B) \sin(t\sqrt{J^2+\omega_B^2})^2}{J^2+\omega_B^2}\,.
	\end{align*}
	Furthermore, we can analytically obtain
	\begin{align*}
		D&\left[\Theta_{SB}(t)||\tau_S\otimes\tau_B\right]=\beta \frac{2J^2 \omega_B \tanh(\beta \omega_B) \sin(t\sqrt{J^2+\omega_B^2})^2}{J^2+\omega_B^2}\,.
	\end{align*}
	Then, we obtain $ \beta \langle \tilde{Q} \rangle_B + D\left[\Theta_{SB} (t)|| \tau_S \otimes \tau_B \right]=0$, which verifies  Eq.~\eqref{eq:JarzynskiSupp}. Interestingly, this examples shows how our approach can well describe the scenario where the quantum ``bath'' (or environment) is small, and thus affected by a large backaction. In this case, even if there is no system energy change, we can still define heat, while the  quantum relative entropy plays the role of work performed by the system onto the bath.

	\subsubsection{4-2. Spin-boson model}
	
	Let us consider the following spin-boson model with the time-independent Hamiltonian~\cite{Schlosshauer}
	\begin{align*}
		H=\frac{\omega_0}{2}\sigma_z+\sum_k \omega_ka_k^{\dagger}a_k+\sigma_z\sum_k(g_ka_k+g_k^* a_k^{\dagger})\,. 
	\end{align*}
	In interaction picture, we obtain
	\begin{align*}
		H(t)=\sigma_z\sum_k(g_k a_k e^{-i\omega_k t}+g_k^* a_k^{\dagger}e^{+i\omega_k t})\,,  
	\end{align*}
	and by the Magnus expansion, the propagator can be simply given by
	\begin{equation}
	U_t=\exp\left[-it (H_0+H_1)\right]\,,
	\label{eq:U}
	\end{equation}
	where the higher terms are vanishing, and $H_0$ and $H_1$ are, respectively, defined as
	\begin{align*}
		\begin{split}
			H_0&\equiv\frac{1}{t}\int_{0}^{t}H(t_1)dt_1\\
			H_1&\equiv-\frac{i}{2t}\int_{0}^tdt_1\int_{0}^{t_1}dt_2 \left[H(t_1),H(t_2)\right]\,.
		\end{split}
	\end{align*}
	Then, we can obtain
	\begin{equation}
	H_0=\sigma_z\sum_{k}\left(G_k(t)a_k-G^*(t)a_k^{\dagger}\right)\,,
	\label{eq:H0}
	\end{equation}
	where 
	\begin{equation}
	G_k(t)\equiv g_k\frac{\sin(\omega_kt/2)}{\omega_kt/2}e^{-i\omega_kt/2}\,.
	\label{eq:G}
	\end{equation}
	Also, $H_1$ is given by
	\begin{equation}
	H_1 =-\sum_k\mathcal{G}_k\,,
	\label{eq:H1}
	\end{equation}
	where 
	\begin{align*}
		\mathcal{G}_k\equiv\frac{|g_k|^2}{\omega_k}\left(1-\frac{\sin(\omega_k t)}{\omega_k t}\right)\,.
	\end{align*}
	From Eq.~\eqref{eq:H0}, Eq.~\eqref{eq:H1} and Eq.~\eqref{eq:U}, the propagator becomes
	\begin{align*}
		U_t=\exp\left[-it\sum_k\left(\sigma_z(G_k(t)a_k+G_k^*(t)a_k^{\dagger})-\mathcal{G}_k\right)\right]\,.
	\end{align*}
	
	Here, we can verify $\langle e^{-\beta \Delta E}\rangle_{\tilde{P}}=1$. $\Delta E$ is defined as
	$\Delta E =\Tr\left[(H_S(t)\otimes\openone_B)U_t(|\epsilon\rangle\langle\epsilon|\otimes\tau_B)U_t^{\dagger}\right]-\epsilon$.
	Due to $H_S(t)=H_S=\frac{\omega_0}{2}\sigma_z$, we can find that $\left[H_S,U_t\right]=0$, which leads to $\Delta E=\langle\epsilon |H_S|\epsilon\rangle-\epsilon=0$ because $|\epsilon\rangle$ is an eigenbasis of $H_S$ corresponding to the eigenvalue $\epsilon$. Therefore, 
	\begin{align*}
		\langle e^{-\beta\Delta E}\rangle_{\tilde{P}}=1\,.
	\end{align*}

	We can also compute the guessed quantum heat $\langle \tilde{Q}\rangle_{B}$, which also corresponds to the negative guessed quantum work $-\langle \tilde{W}\rangle$ in this model. The definition of the guessed quantum heat is $\langle \tilde{Q}\rangle_B=\Tr\left[H_B\tau_B\right]-\Tr\left[H_B\Theta_{SB}(t)\right]$, where
	\begin{align*}
		\Theta_{SB}(t)=\sum_{\epsilon}\frac{e^{-\beta\Tr\left[H_S(t)\Phi_t(|\epsilon\rangle\langle\epsilon|)\right]}}{\tilde{Z}_S(t)}U_t(|\epsilon\rangle\langle\epsilon|\otimes\tau_B)U_t^{\dagger}\,.
	\end{align*}
	Recall that we consider the time-independent Hamiltonian $H_S(t)=H_S$. For the dephasing process, we have $\Phi_t(|\epsilon\rangle\langle\epsilon|)=|\epsilon\rangle\langle\epsilon|$ and $\langle\epsilon|H_S|\epsilon\rangle=\epsilon$.
	In this case, we have $\Theta_{SB}(t)=U_t(\tau_S\otimes\tau_B)U_t^{\dagger}$, which is the exact state of the total system. 
	Then, we have $\Tr\left[H_B\Theta_{SB}(t)\right]=\Tr\left[U_t^{\dagger}H_B U_t(\tau_S\otimes\tau_B)\right]$.
	From the relation $U_t^{\dagger}a_k U_t=a_k+it G_k(t)\sigma_z$ and  $\Tr\left[a_k^{\dagger}\tau_B\right]=\Tr\left[a_k\tau_B\right]=0$,
	we have $\langle \tilde{Q}\rangle_B=\Tr\left[H_B\tau_B\right]-\Tr\left[H_B\Theta_{SB}(t)\right]=-\sum_k\omega_k |G_k(t)|^2t^2$, which from Eq.~\eqref{eq:G} can be explicitly given by
	\begin{equation}
	\langle \tilde{Q}\rangle_B = -\sum_k\omega_k|g_k|^2\left(\frac{\sin(\omega_kt/2)}{\omega_k/2}\right)^2\,.
	\label{eq:guessedheat}
	\end{equation}
	The the noise spectral density is $J(\omega)=\sum_k|g_k|^2\omega\delta(\omega-\omega_k)$; therefore
	\begin{align*}
		\langle \tilde{Q}\rangle_B=-\int_{-\infty}^{\infty}J(\omega)\left(\frac{\sin(\omega_kt/2)}{\omega_k/2}\right)^2d\omega\,.
	\end{align*}
	Since we have $\lim_{t\to\infty}\frac{\sin(\omega t/2)}{\omega/2}=\delta\left(\omega/2\right)=2\delta(\omega)$,
	where we used the relation $\lim_{t\to\infty} t\cdot\frac{\sin(xt)}{xt}=\delta(x)$, we can obtain 
	\begin{align*}
		\lim_{t\to\infty}\langle\tilde{Q}\rangle_B=-\int_{-\infty}^{\infty}4J(\omega)\delta^2(\omega)d\omega=-4J(0)\delta(0)=0\,,
	\end{align*}
	which is consistent with our intuition that when $t\to\infty$ there will be no energy exchange between a small system and a large bath for the dephasing process.

	\section{5. Brief review of quantum Stein's lemma}
	In this section, we briefly introduce quantum Stein's lemma by following~Refs. \cite{Ogawa00, Brandao10} in our scenario. Consider that we prepare $n$ independent and identically distributed copies of $\Theta_{SB}(t)$ and $\tau_S(t)\otimes\tau_B$. We observe two POVM $\{O_n,\openone-O_n\}$ at time $t$ on unknown states. The outcome of $O_n$ concludes that the state is $\Theta_{SB}(t)$, while the outcome of $\openone-O_n$ indicates that the state is $\tau_S(t)\otimes\tau_B$. {Here, the state $\Theta_{SB}(t)$ and $\tau_S(t)\otimes\tau_B$ are seen as the null and alternative hypothesis, respectively. }
	Here, we define $\mathcal{A}_n(O_n)\equiv\Tr\left[\Theta_{SB}^{\otimes n}(t)(\openone-O_n)\right]$ as the type-I error probability that the true state is $\Theta^{\otimes n}_{SB}(t)$ while the POVM outcome indicates $(\tau_S(t)\otimes\tau_B)^{\otimes n}$. 
	We also define $\mathcal{B}_n(O_n)\equiv\Tr\left[(\tau_S(t)\otimes\tau_B)^{\otimes n}O_n\right]$ as the type-I\!I error probability that the true state is $(\tau_S(t)\otimes\tau_B)^{\otimes n}$, while the POVM outcome indicates $\Theta_{SB}^{\otimes n}(t)$. 
	Under the restriction that $\mathcal{A}_n(O_n)$ is upper bounded by a small quantity $\delta$, we consider the minimum type-I\!I error probability $\mathcal{B}_n$ defined as $\mathcal{B}_n\equiv \min_{0\leq O_n\leq \openone}\{\mathcal{B}_{n}(O_n)|\mathcal{A}_{n}(O_n)\leq \delta\}$.
	Then, quantum Stein's lemma~\cite{Ogawa00, Brandao10} states that for $0<\delta<1$ we have the following relation:
	\begin{align*}
		\lim_{n\to \infty}\frac{1}{n}\ln\left(\mathcal{B}_n\right)= - D\left[\Theta_{SB}(t)||\tau_S(t)\otimes\tau_B\right]\,.
	\end{align*}
	Therefore, the quantum relative entropy determines the scaling of the quantum hypothesis testing.

	\section{{6. Relation between the guessed quantum work and exact quantum work} }
	\label{sec:guessedVSexact}
	In this section, we discuss the relation between the guessed quantum work and the exact quantum work, which can be obtained by considering the conventional two-point measurement scheme on both the system and bath.  
	
	\subsubsection{6-1. Standard Jarzynski equality from two-point measurement scheme}
	Let us take the same setup in the main text, and suppose that we can measure the bath. We first locally measure the system and bath, and suppose that we obtained two energy values, $\epsilon$ for the system and $q$ for the bath, so that the postmeasurement state becomes $|\epsilon,q\rangle$. Then, we evolve the total system, and locally measure the system and bath again at time $t$. Suppose that we obtain two energy values, $\epsilon'$ for the system and $q'$ for the bath, so that the postmeasurement state becomes $|\epsilon',q'\rangle$. Then,  the quantum work is defined as the difference in the total energy of the system and bath along the trajectory $(\epsilon,q)\to(\epsilon',q')$. Here, $\epsilon$ and $\epsilon'$ are the energy eigenvalue of the time dependent Hamiltonian of the system $H_S(0)$ and $H_S(t)$, respectively. $q$ and $q'$ are the energy eigenvalue of the time-independent Hamiltonian $H_B$ of the bath. Therefore, the work can be defined as
	\begin{align*}
		W=\left(q'+\epsilon'\right)-\left(q+\epsilon\right)\,.
	\end{align*}
	Therefore, the Jarzynski equality from two-measurement scheme is given by
	\begin{equation}
	\langle e^{-\beta W}\rangle = \sum_{\epsilon,\epsilon'q,q'}\frac{e^{-\beta\epsilon}}{Z_S(0)}\cdot\frac{e^{-\beta q}}{Z_B}\Big|\langle \epsilon',q'|U_t|\epsilon,q\rangle\Big|^2 e^{-\beta\left(q'+\epsilon'-q-\epsilon\right)}=e^{-\beta\Delta F_S}\,,
	\label{eq:standardJarzynski}
	\end{equation}
	where $\Delta F_S \equiv F_S(t)-F_S(0)$ and 
	\begin{align*}
		F_S(t)=-\beta^{-1}\ln Z_S(t)=\beta^{-1} \Tr\left[e^{-\beta H_S(t)}\right]\,.
	\end{align*}
	Also, the expectation of the work is given by
	\begin{equation}
	\langle W\rangle = \Tr\left[\left(H_S(t)\otimes\openone_B+\openone_S\otimes H_B\right)U_t\left(\tau_S(0)\otimes\tau_B\right)U_t^{\dagger}\right]-\left(\Tr\left[H_S(0)\tau_S(0)\right]+\Tr\left[H_B\tau_B\right]\right)\,,
	\label{eq:exactwork}
	\end{equation}
	which is the exact quantum work.
	\subsubsection{6-2. Relation between the guessed quantum work and exact quantum work}
	First, let us consider the relation between the expectation value of the guessed quantum work and the exact work. The guessed quantum work is given by
	\begin{align*}
		\langle\tilde{W}\rangle =& \langle\Delta E\rangle-\langle\tilde{Q}\rangle_B\\
		=&\Tr\left[\left(H_S(t)\otimes\openone_B\right)U_t\left(\tau_S(0)\otimes\tau_B\right)U_t^{\dagger}\right]-\Tr\left[H_S(0)\tau_S(0)\right]-\left(\Tr\left[H_B\tau_B\right]-\Tr\left[(\openone_S\otimes H_B)\Theta_{SB}(t)\right]\right)\,.
	\end{align*}
	From Eq.~\eqref{eq:exactwork}, we can obtain
	\begin{align*}
		\langle \tilde{W}\rangle = \langle W\rangle + \Tr\left[\left(\openone_S\otimes H_B\right)\left(\Theta_{SB}(t)-U_t\left(\tau_S(0)\otimes\tau_B\right)U_t^{\dagger}\right)\right]\,,
	\end{align*}
	which shows that the guessed quantum work and the exact quantum work is different from each other by the energy difference between the reduced state of the bath of the best guessed state and the exact final state. They coincide with each other in the closed quantum systems and when the system undergoes the pure dephasing process, as we have shown in the main text and the examples in this Supplemental Material.

	\subsubsection{6-3. The relation between the modified and standard Jarzynski equality}

	From Eq.~\eqref{eq:standardJarzynski} and Theorem.~1, we have
	\begin{align*}
		\langle e^{-\beta \tilde{W}}\rangle_{\tilde{P}} = \langle e^{-\beta W}\rangle  e^{-D\left[\Theta_{SB}(t)||\tau_S(t)\otimes\tau_B\right]}\,.
	\end{align*}
	By applying Jensen's inequality and the monotonicity of the quantum relative entropy, we can obtain
	\begin{equation}
	\langle e^{-\beta\left(W-\langle\tilde{W}\rangle\right)}\rangle\geq e^{D\left[\Theta_{SB}(t)||\tau_S(t)\otimes\tau_B\right]}\geq e^{D\left[\tilde{\rho}_{S}(t)||\tau_S(t)\right]}\,.
	\label{eq:deviation1}
	\end{equation}
	This means that the average of the deviation of the exact quantum work from the guessed quantum work has the lower bound characterized by the quantum relative entropy $D\left[\tilde{\rho}_{S}(t)||\tau_S(t)\right]$.
	
	Furthermore, since we can also have
	\begin{align*}
		\langle e^{-\beta W}\rangle = 
		\langle e^{-\beta\tilde{W}}\rangle_{\tilde{P}}e^{+D\left[\Theta_{SB}(t)||\tau_S(t)\otimes\tau_B\right]}\,
	\end{align*}
	so that the Jensen's inequality yields the inequality for the deviation of the guessed quantum work from the exact quantum work:
	\begin{equation}
	\langle e^{-\beta\left(\tilde{W}-\langle W\rangle\right)}\rangle_{\tilde{P}} \geq e^{-D\left[\Theta_{SB}(t)||\tau_S(t)\otimes\tau_B\right]}\,.
	\label{eq:deviation2}
	\end{equation}
	Therefore, from Eq.~\eqref{eq:deviation1} and Eq.~\eqref{eq:deviation2}, we can also obtain the following inequality with respect to these two different deviations:
	\begin{align*}
		\langle e^{-\beta\left(W-\langle\tilde{W}\rangle\right)}\rangle \langle e^{-\beta\left(\tilde{W}-\langle W\rangle\right)}\rangle_{\tilde{P}}\geq 1\,.
	\end{align*}

	\section{{7. Modified Jarzynski equality for different initial temperatures of the system and bath}}
	\label{sec:differenttemperature}
	
	Let us consider the scenario that initially the temperatures of the system and bath are different each other. Let $\beta_S$ and $\beta_B$ be the initial temperature of the system and the bath, respectively, and let us define the temperature difference as
	\begin{align*}
		\Delta \beta\equiv \beta_B-\beta_S\,.
	\end{align*}
	Then, we can obtain the following Theorem.~\ref{theorem2}, which can be regarded as the extension of Jarzynski-W\'{o}jcik scenario~\cite{Jarzynski04} from the one-time measurement scheme under the restriction that the bath is inaccessible.
	\begin{theorem}
		\label{theorem2}
		When the initial temperatures of the system and bath are different from each other, 
		The Jarzynski equality for the guessed quantum work is
		\begin{align*}
			\langle e^{-\beta_S \tilde{W}}\rangle_{\tilde{P}}=e^{-\beta_S\Delta F_S}e^{-D\left[\Theta_{SB}(t)||\tau_S(t)\otimes\tau_B\right]}e^{-\Delta\beta \langle\tilde{Q}\rangle_B}\,,  
		\end{align*}
		which also yields the following principle of maximum guessed quantum work: 
		\begin{align*}
			\langle\tilde{W}\rangle \geq \Delta F_S+\beta^{-1}D\left[\tilde{\rho}_S(t)||\tau_S(t)\right]+\frac{\Delta\beta}{\beta_S}\langle\tilde{Q}\rangle_B\,,
		\end{align*}
		where 
		\begin{align*}
			\tilde{\rho}_S(t)\equiv\Tr_B\left[\Theta_{SB}(t)\right]=\sum_{\epsilon}\frac{e^{-\beta_S\Tr\left[H_S(t)\Phi_t\left(|\epsilon\rangle\langle\epsilon|\right)\right]}}{\tilde{Z}_S(t)}\Phi_t\left(|\epsilon\rangle\langle\epsilon|\right)\,.
		\end{align*}
	\end{theorem}
	\begin{proof}
		The proof is same to the one in Sec.~2.
		In this case, the initial state is given by
		\begin{align*}
			\tau_S(0)\otimes\tau_B=\frac{e^{-\beta_S H_S(0)}}{Z_S(0)}\otimes \frac{e^{-\beta_B H_B}}{Z_B}\,.
		\end{align*}
		Then, for the internal energy difference, we can obtain
		\begin{align*}
			\langle e^{-\beta_S\Delta E}\rangle_{\tilde{P}}=e^{-\beta_S\Delta F_S}\frac{\tilde{Z}_S(t)}{Z_S(t)}\,,    
		\end{align*}
		where
		\begin{align*}
			\tilde{Z}_S(t)&=\sum_{\epsilon}e^{-\beta_S\Tr\left[H_S(t)\Phi_t\left(|\epsilon\rangle\langle\epsilon|\right)\right]}\\
			Z_S(t)&=\Tr\left[e^{-\beta_S H_S(t)}\right]\,,
		\end{align*}
		where
		\begin{align*}
			\Phi_t\left(|\epsilon\rangle\langle\epsilon|\right)=\Tr_B\left[U_t\left(|\epsilon\rangle\langle\epsilon|\otimes\tau_B\right)U_t^{\dagger}\right]\,.
		\end{align*}
		In this case, the best guessed state is given by
		\begin{align*}
			\Theta_{SB}(t)=\sum_{\epsilon}\frac{e^{-\beta_S\Tr\left[H_S(t)\Phi_t\left(|\epsilon\rangle\langle\epsilon|\right)\right]}}{\tilde{Z}_S(t)}U_t\left(|\epsilon\rangle\langle\epsilon|\otimes\tau_B\right)U_t^{\dagger}\,.
		\end{align*}
		Then, we can obtain
		\begin{align*}
			\Tr\left[\Theta_{SB}(t)\ln\Theta_{SB}(t)\right] = -\ln\tilde{Z}_S(t)-\ln Z_B-\beta_S\Tr\left[\left(H_S(t)\otimes\openone_B\right)\Theta_{SB}(t)\right]-\beta_B\Tr\left[H_B\tau_B\right]\,.
		\end{align*}
		Also, we have
		\begin{align*}
			\Tr\left[\Theta_{SB}(t)\ln\left(\tau_S(t)\otimes\tau_B\right)\right] = -\ln Z_S(t)-\ln Z_B-\beta_S\Tr\left[\left(H_S(t)\otimes\openone_B\right)\Theta_{SB}(t)\right]-\beta_B\Tr\left[\left(\openone_S\otimes H_B\right)\Theta_{SB}(t)\right]\,.
		\end{align*}
		Therefore, we can obtain 
		\begin{align*}
			\frac{\tilde{Z}_S(t)}{Z_S(t)} &= e^{-D\left[\Theta_{SB}(t)||\tau_S(t)\otimes\tau_B\right]}e^{-\beta_B\langle\tilde{Q}\rangle_B}\\
			&=e^{-D\left[\Theta_{SB}(t)||\tau_S(t)\otimes\tau_B\right]}e^{-\beta_S\langle\tilde{Q}\rangle_B}e^{-\Delta \beta\langle\tilde{Q}\rangle_B}
		\end{align*}
		where
		\begin{equation}
		\langle\tilde{Q}\rangle_B = \Tr\left[H_B\tau_B\right]-\Tr\left[(\openone_S\otimes H_B)\Theta_{SB}(t)\right]\,,
		\label{eq:ratiopartition}
		\end{equation}
		which is the guessed quantum heat. By definition of the guessed quantum work:
		\begin{align*}
			\tilde{W}\equiv \Delta E-\langle\tilde{Q}\rangle_B\,,
		\end{align*}
		we can obtain
		\begin{align*}
			\langle e^{-\beta_S \tilde{W}}\rangle_{\tilde{P}}=e^{-\beta_S\Delta F_S}e^{-D\left[\Theta_{SB}(t)||\tau_S(t)\otimes\tau_B\right]}e^{-\Delta\beta \langle\tilde{Q}\rangle_B}\,.   
		\end{align*}
		Applying Jensen's inequality and the monotonicity of the quantum relative entropy, we can obtain
		\begin{align*}
			\langle\tilde{W}\rangle \geq \Delta F_S+\beta_S^{-1}D\left[\tilde{\rho}_S(t)||\tau_S(t)\right]+\frac{\Delta\beta}{\beta_S}\langle\tilde{Q}\rangle_B\,,
		\end{align*}
		where 
		\begin{align*}
			\tilde{\rho}_S(t)\equiv\Tr_B\left[\Theta_{SB}(t)\right]=\sum_{\epsilon}\frac{e^{-\beta_S\Tr\left[H_S(t)\Phi_t\left(|\epsilon\rangle\langle\epsilon|\right)\right]}}{\tilde{Z}_S(t)}\Phi_t\left(|\epsilon\rangle\langle\epsilon|\right)\,.
		\end{align*}
	\end{proof}

\end{document}